\documentclass{amsproc}%
\usepackage{amsfonts}
\usepackage{amsmath}
\usepackage{amssymb}
\usepackage{graphicx}%
\theoremstyle{plain}
\newtheorem{acknowledgement}{Acknowledgement}

\newtheorem{corollary}{Corollary}

\newtheorem{definition}{Definition}

\newtheorem{lemma}{Lemma}

\newtheorem{proposition}{Proposition}
\newtheorem{remark}{Remark}

\newtheorem{theorem}{Theorem}
\numberwithin{equation}{section}
\begin{document}
\title[Non Archimedean Klein-Gordon Equations]{The Cauchy Problem for Non-Archimedean Pseudodifferential Equations of
Klein-Gordon Type}
\author{W. A. Z\'{u}\~{n}iga-Galindo}
\address{Centro de Investigaci\'{o}n y de Estudios Avanzados del Instituto
Polit\'{e}cnico Nacional\\
Departamento de Matem\'{a}ticas- Unidad Quer\'{e}taro\\
Libramiento Norponiente \#2000, Fracc. Real de Juriquilla. Santiago de
Quer\'{e}taro, Qro. 76230\\
M\'{e}xico}
\email{wazuniga@math.cinvestav.edu.mx}
\thanks{The author was partially supported by Conacyt (Mexico), Grant \# 127794.}

\begin{abstract}
In this article we introduce a new class of non-Archime\-dean
pseudodifferential equations of Klein-Gordon type and study the corresponding
Cauchy problem for these equations. A remarkable fact is that the
non-Archime\-dean Klein-Gordon equations exhibit many similar properties to
the classical Klein-Gordon equations.

\end{abstract}
\subjclass{35S05, 26E30}
\keywords{Pseudodifferential operators, Klein-Gordon operators, Cauchy problems,\ $p$%
-adic analysis. }
\maketitle

\section{Introduction}

In the 1980's I. Volovich proposed that the world geometry in regimes smaller
than the Planck scale might be non-Archimedean \ \cite{Vol1}, \cite{Vol2}.
This hypothesis conducts naturally to consider models involving geometry and
analysis over $\mathbb{Q}_{p}$, the field of $p$-adic numbers, $p$ being a
rational prime $\geq2$. Since then, a big number of articles have appeared
exploring these and related themes, see e.g. \cite{D-K-K-V}, \cite[Chapter
6]{Var} and the references therein. In particular, nowadays there is a strong
interest in studying pseudodifferential equations over $p$-adic fields, see
e.g. \cite{A-K-S}, \cite{Koch}, \cite{Koch2}, \cite{K-Kos}.

In this article we introduce a new class of non-Archimedean pseudodifferential
equations of Klein-Gordon type. We work on the $p$-adic Minkowski space which
is the quadratic space $\left(  \mathbb{Q}_{p}^{4},Q\right)  $ where
$Q(k)=k_{0}^{2}-k_{1}^{2}-k_{2}^{2}-k_{3}^{2}$. Our starting point is a result
of Rallis-Schiffmann that asserts the existence of a unique measure on
$V_{t}=\left\{  k\in\mathbb{Q}_{p}^{4}:Q\left(  k\right)  =t\right\}  $ which
is invariant under the orthogonal group $O(Q)$\ of $Q$, see Proposition
\ref{prop1} or \cite{R-S}. By using Gel'fand-Leray differential forms, we
reformulate this results in terms of Dirac distributions $\delta\left(
Q\left(  k\right)  -t\right)  $ invariant under $O(Q)$, see Remark \ref{nota1}
and Lemma \ref{lemma1}. We introduce the positive and negative mass shells
$V_{m^{2}}^{+}$ and $V_{m^{2}}^{-}$, here $m$ is the `mass parameter' which is
taken to be a nonzero $p$-adic number. The restriction of $\delta\left(
Q\left(  k\right)  -t\right)  $ to $V_{m^{2}}^{\pm}$ gives two distributions
$\delta_{\pm}\left(  Q\left(  k\right)  -t\right)  $ which are invariant under
$\mathcal{L}_{+}^{\uparrow}$, the `Lorentz proper group', see Definition
\ref{Lorentz group}, and that satisfy $\delta\left(  Q\left(  k\right)
-t\right)  =\delta_{+}\left(  Q\left(  k\right)  -t\right)  +\delta_{-}\left(
Q\left(  k\right)  -t\right)  $, see Lemma \ref{lemma4}. The $p$-adic
Klein-Gordon type pseudodifferential operators introduced here have the form
\[
\left(  \square_{\alpha,m}\varphi\right)  \left(  x\right)  =\mathcal{F}%
_{k\rightarrow x}^{-1}\left[  \left\vert Q(k)-m^{2}\right\vert _{p}^{\alpha
}\mathcal{F}_{x\rightarrow k}\varphi\right]  \text{, }\alpha>0\text{, }%
m\in\mathbb{Q}_{p}\smallsetminus\left\{  0\right\}  \text{,}%
\]
where $\mathcal{F}$ denotes the Fourier-Minkowski transform. We solve the
Cauchy problem for these operators, see Theorem \ref{TheoA}.

The equations $\left(  \square_{\alpha,m}\phi\right)  \left(  t,\boldsymbol{x}%
\right)  =0$ have many similar properties to the classical Klein-Gordon
equations, see e.g. \cite{DJager}, \cite{DJager2}, \cite{Schweber}. These
equations admit plane waves as weak solutions, see Lemma
\ref{Lemma_plane_waves}; the distributions $a\mathcal{F}^{-1}\left[
\delta_{+}\left(  Q\left(  k\right)  -m^{2}\right)  \right]  +b\mathcal{F}%
^{-1}\left[  \delta_{-}\left(  Q\left(  k\right)  -m^{2}\right)  \right]  $,
$a,b\in\mathbb{C}$, are weak solutions of these equations, see Proposition
\ref{Prop1}. The locally constant functions
\begin{equation}
\phi\left(  t,\boldsymbol{x}\right)  =%
{\displaystyle\int\limits_{U_{Q,m}}}
\chi_{p}\left(  \boldsymbol{x}\cdot\boldsymbol{k}\right)  \left\{  \chi
_{p}\left(  -t\omega\left(  \boldsymbol{k}\right)  \right)  \phi_{+}\left(
\boldsymbol{k}\right)  +\chi_{p}\left(  t\omega\left(  \boldsymbol{k}\right)
\right)  \phi_{-}\left(  \boldsymbol{k}\right)  \right\}  d^{3}\boldsymbol{k},
\label{phi}%
\end{equation}
where $\chi_{p}\left(  \cdot\right)  $ denotes the standard additive character
of $\mathbb{Q}_{p}$ and $\phi_{\pm}$ are locally constant functions with
support in $U_{Q,m}$, are weak solutions of these equations.

At this point, it is relevant to mention that the operators, equations and
techniques introduced here are new. In \cite{A-K-S} a very general theory for
pseudodifferential operators and equations involving symbols that vanish only
at the origin was developed. This theory cannot be applied here because our
symbols ($\left\vert Q(k)-m^{2}\right\vert _{p}^{\alpha}$) have infinitely
many zeros.

Finally, the quantization of solutions (\ref{phi}) and its connection with the
second quantization is an open problem. It `seems' that the construction of a
neutral (real) quantum scalar field with mass parameter $m\in\mathbb{Q}%
_{p}^{\times}$ can be carried out using the machinery of the second
\ quantization starting with $\mathcal{H}=L^{2}\left(  V_{m^{2}}^{+}%
,d\lambda_{m^{2}}\right)  $, the state space for a single spin-zero particle
of mass $m$, see e.g. \cite{Dim}, \cite{Folland}, \cite{Reed-Simon}. But there
are several mathematical and physical obstacles to overcome.

\section{Preliminaries}

Along this article $p$ will denote a prime number different from $2$. The
field of $p-$adic numbers $\mathbb{Q}_{p}$ is defined as the completion of the
field of rational numbers $\mathbb{Q}$ with respect to the $p-$adic norm
$|\cdot|_{p}$, which is defined as
\[
|x|_{p}=%
\begin{cases}
0 & \text{if }x=0\\
p^{-\gamma} & \text{if }x=p^{\gamma}\dfrac{a}{b},
\end{cases}
\]
where $a$ and $b$ are integers coprime with $p$. The integer $\gamma:=ord(x)$,
with $ord(0):=+\infty$, is called the\textit{ }$p-$\textit{adic order of} $x$.
Any $p-$adic number $x\neq0$ has a unique expansion $x=p^{ord(x)}\sum
_{j=0}^{\infty}x_{i}p^{j}$, where $x_{j}\in\{0,1,2,\dots,p-1\}$ and $x_{0}%
\neq0$. Thus any nonzero $p-$adic number can be written uniquely as
$x=p^{ord(x)}ac(x)$, where $ac(x)$, \textit{the angular component of }$x$, is
a unit i.e. $\left\vert ac(x)\right\vert _{p}=1$. For a unit $a=\sum
_{j=0}^{\infty}a_{i}p^{j}$, $a_{0}\neq0$, we define $\overline{a}:=a_{0}%
\in\mathbb{F}_{p}$, where $\mathbb{F}_{p}$ denotes the field of $p$ elements.
We also define \textit{the fractional part of }$x\in\mathbb{Q}_{p}$, denoted
$\{x\}_{p}$, as the rational number
\[
\{x\}_{p}=%
\begin{cases}
0 & \text{if }x=0\text{ or }ord(x)\geq0\\
p^{\text{ord}(x)}\sum_{j=0}^{-ord(x)-1}x_{j}p^{j} & \text{if }ord(x)<0.
\end{cases}
\]
Set $\chi_{p}(y)=\exp(2\pi i\{y\}_{p})$ for $y\in\mathbb{Q}_{p}$. The map
$\chi_{p}(\cdot)$ is an additive character on $\mathbb{Q}_{p}$, i.e. a
continuous map from $\mathbb{Q}_{p}$ into the unit circle satisfying $\chi
_{p}(y_{0}+y_{1})=\chi_{p}(y_{0})\chi_{p}(y_{1})$, $y_{0},y_{1}\in
\mathbb{Q}_{p}$.

We extend the $p-$adic norm to $\mathbb{Q}_{p}^{n}$ by taking%
\[
||x||_{p}:=\max_{1\leq i\leq n}|x_{i}|_{p},\qquad\text{for }x=(x_{1}%
,\dots,x_{n})\in\mathbb{Q}_{p}^{n}.
\]
For $\gamma\in\mathbb{Z}$, denote by $B_{\gamma}^{n}(a)=\{x\in\mathbb{Q}%
_{p}^{n}:||x-a||_{p}\leq p^{\gamma}\}$ \textit{the ball of radius }$p^{\gamma
}$ \textit{with center at} $a=(a_{1},\dots,a_{n})\in\mathbb{Q}_{p}^{n}$, and
take $B_{\gamma}^{n}(0):=B_{\gamma}^{n}$. Note that $B_{\gamma}^{n}%
(a)=B_{\gamma}(a_{1})\times\cdots\times B_{\gamma}(a_{n})$, where $B_{\gamma
}(a_{i}):=\{x\in\mathbb{Q}_{p}:|x_{i}-a_{i}|_{p}\leq p^{\gamma}\}$ is the
one-dimensional ball of radius $p^{\gamma}$ with center at $a_{i}\in
\mathbb{Q}_{p}$. The ball $B_{0}^{n}(0)$ is equal to the product of $n$ copies
of $B_{0}(0):=\mathbb{Z}_{p}$, \textit{the ring of }$p-$\textit{adic integers}.

A complex-valued function $\varphi$ defined on $\mathbb{Q}_{p}^{n}$ is
\textit{called locally constant} if for any $x\in\mathbb{Q}_{p}^{n}$ there
exists an integer $l(x)\in\mathbb{Z}$ such that%
\begin{equation}
\varphi(x+x^{\prime})=\varphi(x)\text{ for }x^{\prime}\in B_{l(x)}^{n}.
\label{local_constancy}%
\end{equation}
A function $\varphi:\mathbb{Q}_{p}^{n}\rightarrow\mathbb{C}$ is called a
\textit{Bruhat-Schwartz function (or a test function)} if it is locally
constant with compact support. The $\mathbb{C}$-vector space of
Bruhat-Schwartz functions is denoted by $\mathbf{S}(\mathbb{Q}_{p}^{n})$. Let
$\mathbf{S}^{\prime}(\mathbb{Q}_{p}^{n})$ denote the set of all functionals
(distributions) on $\mathbf{S}(\mathbb{Q}_{p}^{n})$. All functionals on
$\mathbf{S}(\mathbb{Q}_{p}^{n})$ are continuous, see e.g. \cite[p. 84]{V-V-Z}.

For a detailed discussion on $p$-adic analysis the reader may consult
\cite{A-K-S}, \cite{Koch}, \cite{Taibleson}, \cite{V-V-Z}.

\subsection{Fourier transform on finite dimensional vector spaces}

Let $E$ be a finite dimensional vector space over $\mathbb{Q}_{p}$ and
$\chi_{p}$ a non-trivial additive character of $\mathbb{Q}_{p}$\ as before.
Let $\left[  x,y\right]  $ be a symmetric non-degenerate $\mathbb{Q}_{p}%
-$bilinear form on $E\times E$. Thus $Q(e):=\left[  e,e\right]  $, $e\in E$ is
a \textit{non-degenerate quadratic form} on $E$.\ We identify $E$ with its
algebraic dual $E^{\ast}$ by means of $\left[  \cdot,\cdot\right]  $. We now
identify the dual group (i.e. the Pontryagin dual) of $\left(  E,+\right)  $
with $E^{\ast}$ by taking $\left\langle e,e^{\ast}\right\rangle =\chi
_{p}\left(  \left[  e,e^{\ast}\right]  \right)  $ where $\left[  e,e^{\ast
}\right]  $ is the algebraic duality. The Fourier transform takes the form%
\[
\widehat{\varphi}\left(  y\right)  =%
{\displaystyle\int\limits_{E}}
\varphi\left(  x\right)  \chi_{p}\left(  \left[  x,y\right]  \right)  dx\text{
for }\varphi\in L^{1}\left(  E\right)  \text{,}%
\]
where $dx$ is a Haar measure on $E$.

Let $\mathcal{L}\left(  E\right)  $ be the space of continuous functions
$\varphi$ in $L^{1}(E)$ whose Fourier transform $\widehat{\varphi}$ is in
$L^{1}(E)$.The measure $dx$ can be normalized uniquely in such manner that
$\widehat{\left(  \widehat{\varphi}\right)  }\left(  x\right)  =\varphi\left(
-x\right)  $ for every $\varphi$ belonging to $\mathcal{L}\left(  E\right)  $.
We say that $dx$ is \textit{\ a self-dual measure relative\ to} $\chi
_{p}\left(  \left[  \cdot,\cdot\right]  \right)  $.

For further details about the material presented in this section the reader
may consult \cite{We1}.

\subsection{The $p$-adic Minkowski space}

We take $E$ to be the $\mathbb{Q}_{p}$-vector space of dimension $4$. By
fixing a basis we identify $E$ with $\mathbb{Q}_{p}^{4}$ considered as a
$\mathbb{Q}_{p}$-vector space. For $x=\left(  x_{0},x_{1},x_{2},x_{3}\right)
:=\left(  x_{0},\boldsymbol{x}\right)  $ and $y=\left(  y_{0},y_{1}%
,y_{2},y_{3}\right)  :=\left(  y_{0},\boldsymbol{y}\right)  $ in
$\mathbb{Q}_{p}^{4}$ we set
\begin{equation}
\left[  x,y\right]  :=x_{0}y_{0}-x_{1}y_{1}-x_{2}y_{2}-x_{3}y_{3}:=x_{0}%
y_{0}-\boldsymbol{x}\cdot\boldsymbol{y}, \label{bilinear_form}%
\end{equation}
which is a symmetric non-degenerate bilinear form. \textit{From now on, we use
}$\left[  x,y\right]  $\textit{ to mean bilinear form (\ref{bilinear_form})}.
Then $\left(  \mathbb{Q}_{p}^{4},Q\right)  $, with $Q(x)=\left[  x,x\right]  $
is a \textit{quadratic vector space} and $Q$ is a non-degenerate quadratic
form on $\mathbb{Q}_{p}^{4}$. We will call $\left(  \mathbb{Q}_{p}%
^{4},Q\right)  $ \textit{the }$p$\textit{-adic Minkowski space}.

On $\left(  \mathbb{Q}_{p}^{4},Q\right)  $, the Fourier transform takes the
form%
\begin{equation}
\mathcal{F}\left[  \varphi\right]  \left(  k\right)  =%
{\displaystyle\int\limits_{\mathbb{Q}_{p}^{4}}}
\chi_{p}\left(  \left[  x,k\right]  \right)  \varphi\left(  x\right)
d^{4}x\text{ for }\varphi\in L^{1}\left(  \mathbb{Q}_{p}^{4}\right)  \text{,}
\label{Lorentz_Fourier transform}%
\end{equation}
where $d^{4}x$ is a self-dual measure for $\chi_{p}\left(  \left[  \cdot
,\cdot\right]  \right)  $, i.e. $\mathcal{F}\left[  \mathcal{F}\left[
\varphi\right]  \right]  \left(  x\right)  =\varphi\left(  -x\right)  $ for
every $\varphi$ belonging to $\mathcal{L}\left(  \mathbb{Q}_{p}^{4}\right)  $.
Notice that $d^{4}x$ is equal to a positive multiple of the normalized Haar
measure on $\mathbb{Q}_{p}^{4}$, i.e. $d^{4}x=Cd^{4}\mu\left(  x\right)  $,
with $C>0$.

\begin{remark}
(i) We set the usual Fourier transform $\mathfrak{F}$ to be
\[
\mathfrak{F}\left[  \varphi\right]  \left[  k\right]  :=%
{\displaystyle\int\limits_{\mathbb{Q}_{p}^{4}}}
\chi_{p}\left(  x_{0}k_{0}+x_{1}k_{1}+x_{2}k_{2}+x_{3}k_{3}\right)
\varphi\left(  x\right)  d^{4}\mu\left(  x\right)  \text{ for }\varphi\in
L^{1}\left(  \mathbb{Q}_{p}^{4}\right)  \text{,}%
\]
where $d^{4}\mu\left(  x\right)  $ is the normalized Haar measure of
$\mathbb{Q}_{p}^{4}$. The connection between $\mathcal{F}$ and $\mathfrak{F}$
is given by the formula%
\[
\mathfrak{F}\left[  \mathcal{F}\left[  \varphi\left(  x_{0},x_{1},x_{2}%
,x_{3}\right)  \right]  \right]  =C\varphi\left(  x_{0},-x_{1},-x_{2}%
,-x_{3}\right)  ,
\]
which is equally valid for integrable functions as well as distributions.

(ii) Note that $C=1$, i.e. $d^{4}x=d^{4}\mu\left(  x\right)  $. Indeed, take
$\varphi\left(  x\right)  $ to be the characteristic function of
$\mathbb{Z}_{p}^{4}$, now
\begin{align*}
\varphi\left(  x\right)   &  =\mathcal{F}\left[  \mathcal{F}\left[
\varphi\right]  \right]  \left(  x\right)  =C\mathcal{F}_{k\rightarrow
x}\left[  \mathfrak{F}\left[  \varphi\right]  \left(  k_{0},-\boldsymbol{k}%
\right)  \right]  =C\mathcal{F}_{k\rightarrow x}\left[  \varphi\left(
k_{0},\boldsymbol{k}\right)  \right] \\
&  =C^{2}\mathfrak{F}\left[  \varphi\left(  x_{0},-\boldsymbol{x}\right)
\right]  =C^{2}\varphi\left(  x\right)  \text{,}%
\end{align*}
therefore $C=\pm1$.
\end{remark}

\subsection{Invariant measures under the orthogonal group $O(Q)$}

We set $Q(x)=\left[  x,x\right]  $ as before. We also set%
\[
G:=\left[
\begin{array}
[c]{cccc}%
1 & 0 & 0 & 0\\
0 & -1 & 0 & 0\\
0 & 0 & -1 & 0\\
0 & 0 & 0 & -1
\end{array}
\right]  .
\]
Then $Q(x)=x^{T}Gx$, where $T$ denotes the transpose of a matrix. The
orthogonal group of $Q(x)$\ is defined as
\begin{align*}
O(Q)  &  =\left\{  \Lambda\in GL_{4}\left(  \mathbb{Q}_{p}\right)  :\left[
\Lambda x,\Lambda y\right]  =\left[  x,y\right]  \right\} \\
&  =\left\{  \Lambda\in GL_{4}\left(  \mathbb{Q}_{p}\right)  :\Lambda
^{T}G\Lambda=G\right\}  .
\end{align*}

Notice that any $\Lambda\in O(Q)$ satisfies $\det\Lambda=\pm1$. We consider
$O(Q)$ as a $p$-adic Lie subgroup of $GL_{4}\left(  \mathbb{Q}_{p}\right)  $
which is a $p$-adic Lie group.

For $t\in\mathbb{Q}_{p}^{\times}$, we set%
\[
V_{t}:=\left\{  k\in\mathbb{Q}_{p}^{4}:Q\left(  k\right)  =t\right\}  .
\]

\begin{proposition}
[{Rallis-Schiffman, \cite[Proposition 2-2]{R-S}}]\label{prop1}The orthogonal
group $O\left(  Q\right)  $ acts transitively on $V_{t}$. On each orbit
$V_{t}$ there is a measure which is invariant under $O\left(  Q\right)  $ and
unique up to multiplication by a positive constant.
\end{proposition}

For each $t\in\mathbb{Q}_{p}^{\times}$, let $d\mu_{t}$ be a measure on $V_{t}$
invariant under $O\left(  Q\right)  $. Since $V_{t}$\ is closed in
$\mathbb{Q}_{p}^{4}$, it is possible to consider $d\mu_{t}$ as a measure on
$\mathbb{Q}_{p}^{4}$ supported on $V_{t}$.

\subsection{$p$-adic analytic manifolds}

We give a brief review on $p$-adic manifolds in the sense of Serre. For
further details the reader may consult \cite{Igusa}.

We recall that a $p$\textit{-adic analytic manifold }$X$\textit{ of dimension
}$n$ is a topological Hausdorff space equipped with an atlas\textit{, }which
is a family of compatible charts $\left\{  \left(  U,\phi_{U}\right)
\right\}  $ covering $X$. As in the Archimedean case, each chart is a pair
$\left(  U,\phi_{U}\right)  $, where $U$ is a nonempty open subset of $X$ and
$\phi_{U}$ is a homeomorphism from $U$ to $\phi_{U}\left(  U\right)
\subset\mathbb{Q}_{p}^{n}$, with $n$ fixed.

\subsubsection{Gel'fand-Leray differential forms}

Since $\nabla Q(k)\neq0$ for any $k\in V_{t}$, by using the non-Archimedean
implicit function theorem one verifies that $V_{t}$ is \textit{a }$p$-adic
closed submanifold of codimension\textit{ }$1$.

The condition $\nabla Q(k)\neq0$ for any $k\in V_{t}$, implies the existence
of a $p$-adic analytic differential form $\lambda_{t}$ on $V_{t}$ satisfying%
\begin{equation}
dk_{0}\wedge dk_{1}\wedge dk_{2}\wedge dk_{3}=dQ\left(  k\right)
\wedge\lambda_{t}. \label{Gelfand-Leray}%
\end{equation}
A such form is typically called a \textit{Gel'fand-Leray form. }The
differential form $\lambda_{t}$ is not unique but its restriction to $V_{t}$
is independent of the choice of $\lambda_{t}$, see e.g. \cite[Chap. III, Sect.
1-9]{G-S}, \cite[Section 7.4 and 7.6]{Igusa}, \cite{Z-G4}. We denote the
corresponding measure as $\lambda_{t}\left(  A\right)  =%
{\textstyle\int\nolimits_{A}}
d\lambda_{t}$ for an open compact subset $A$ of $V_{t}$.

The notation $\left(  k_{0},\ldots,\widehat{k}_{l\left(  j\right)  }%
,\ldots,k_{3}\right)  $ means omit the $l\left(  j\right)  $th coordinate. We
now describe this measure in a suitable chart. We may assume that $V_{t}$ is a
countable disjoint union of submanifolds of the form%
\begin{equation}
V_{t}^{\left(  j\right)  }:=\left\{
\begin{array}
[c]{l}%
\left(  k_{0},\ldots,k_{3}\right)  \in\mathbb{Q}_{p}^{4}:k_{l\left(  j\right)
}=h_{j}\left(  k_{0},\ldots,\widehat{k}_{l\left(  j\right)  },\ldots
,k_{3}\right) \\
\text{with }\left(  k_{0},\ldots,\widehat{k}_{l\left(  j\right)  }%
,\ldots,k_{3}\right)  \in V_{t}^{\left(  j\right)  },
\end{array}
\right\}  \label{Mj}%
\end{equation}
where $h_{j}\left(  k_{0},\ldots,\widehat{k}_{l\left(  j\right)  }%
,\ldots,k_{3}\right)  $ is a $p$-adic analytic function on some open compact
subset $V_{j}$\ of $\mathbb{Q}_{p}^{3}$, and $\frac{\partial Q}{\partial
k_{l\left(  j\right)  }}\left(  z\right)  \neq0$ for any $z\in V_{t}^{\left(
j\right)  }$. If $A$ is a compact open subset contained in $V_{t}^{\left(
j\right)  }$, then%
\begin{equation}
\lambda_{t}\left(  A\right)  =%
{\displaystyle\int\limits_{h_{j}^{-1}\left(  A\right)  }}
\frac{dk_{0}\ldots d\widehat{k}_{l\left(  j\right)  }\ldots dk_{3}}{\left\vert
\frac{\partial Q}{\partial k_{l\left(  j\right)  }}\left(  k\right)
\right\vert _{p}}, \label{Omega}%
\end{equation}
where we are identifying the set $A\subset V_{t}^{\left(  j\right)  }$ with
\ the set of all the coordinates of the points of $A$, which is a subset of
$\mathbb{Q}_{p}^{3}$, and $h_{j}^{-1}\left(  A\right)  $ denotes the subset of
$\mathbb{Q}_{p}^{4}$ consisting of the points $(k_{0},k_{1},k_{2},k_{3})$ such
that \ $k_{l\left(  j\right)  }=h_{j}\left(  k_{0},\ldots,\widehat
{k}_{l\left(  j\right)  },\ldots,k_{3}\right)  $ for $\left(  k_{0}%
,\ldots,\widehat{k}_{l\left(  j\right)  },\ldots,k_{3}\right)  \in A$.

\begin{remark}
\label{nota1}(i) Let $\mathcal{T}\left(  V_{t}\right)  $ denote the family of
all compact open subsets of $V_{t}$. Then $\lambda_{t}$ is a additive function
on $\mathcal{T}\left(  V_{t}\right)  $ such that $\lambda_{t}\left(  A\right)
\geq0$ for every $A$ in $\mathcal{T}\left(  V_{t}\right)  $. By
Carath\'{e}odory's extension theorem $\lambda_{t}$ has a unique extension to
the $\sigma$-algebra generated by $\mathcal{T}\left(  V_{t}\right)  $.\ We
also note that the measure $\lambda_{t}$ is supported on $V_{t}$.

(ii) Let $\mathbf{S}\left(  V_{t}\right)  $ denote the $\mathbb{C}$-vector
space generated by the characteristic functions of the elements of
$\mathcal{T}\left(  V_{t}\right)  $. The fact that $\lambda_{t}$ is a positive
additive function on $\mathcal{T}\left(  V_{t}\right)  $ is equivalent to say
that
\[%
\begin{array}
[c]{ccc}%
\mathbf{S}\left(  V_{t}\right)  & \rightarrow & \mathbb{C}\\
&  & \\
\varphi & \rightarrow &
{\displaystyle\int\limits_{\mathbb{Q}_{p}^{4}}}
\varphi\left(  k\right)  d\lambda_{t}\left(  k\right)
\end{array}
\]
is a positive distribution. We can identify the measure $d\lambda_{t}$ with a
distribution on $\mathbb{Q}_{p}^{4}$ supported on $V_{t}$.

(iv) Some authors use $\delta\left(  Q\left(  k\right)  -t\right)  $ or
$\delta\left(  Q\left(  k\right)  -t\right)  d^{4}k$ to denote the measure
$d\lambda_{t}$. We will use $\delta\left(  Q\left(  k\right)  -t\right)  $.
\end{remark}

\begin{remark}
Let $G_{0}$ be a subgroup of $GL_{4}(\mathbb{Q}_{p})$. Let $\varphi
\in\boldsymbol{S}\left(  \mathbb{Q}_{p}^{4}\right)  $ and let $\Lambda\in
G_{0}$. We define the action of $\Lambda$ on $\varphi$ by putting%
\[
\left(  \Lambda\varphi\right)  \left(  x\right)  =\varphi\left(  \Lambda
^{-1}x\right)  ,
\]
and the action of $\Lambda$ on a distribution $T\in\boldsymbol{S}^{\prime
}\left(  \mathbb{Q}_{p}^{4}\right)  $ by putting%
\[
\left(  \Lambda T,\varphi\right)  =\left(  T,\Lambda^{-1}\varphi\right)  .
\]
We say that $T$ is invariant under $G_{0}$ if $\Lambda T=T$ for any
$\Lambda\in G_{0}$.
\end{remark}

\begin{lemma}
\label{lemma1}With the above notation, we have $d\mu_{t}=Ad\lambda_{t}$ for
some positive constant $A$.
\end{lemma}

\begin{proof}
By Remark \ref{nota1} and Proposition \ref{prop1}, it is sufficient to show
that the distribution $\delta\left(  Q\left(  k\right)  -t\right)  $ is
invariant under $O\left(  Q\right)  $, i.e.%
\[%
{\displaystyle\int\limits_{V_{t}}}
\varphi\left(  \Lambda k\right)  d\lambda_{t}\left(  k\right)  =%
{\displaystyle\int\limits_{V_{t}}}
\varphi\left(  k\right)  d\lambda_{t}\left(  k\right)
\]
for any $\Lambda\in O(Q)$ and $\varphi\in\mathbf{S}\left(  \mathbb{Q}_{p}%
^{4}\right)  $. Since $V_{t}$ is invariant under $\Lambda$, it is sufficient
to show that $d\lambda_{t}\left(  k\right)  =d\lambda_{t}\left(  y\right)  $
under $k=\Lambda^{-1}y$, for any $\Lambda\in O(Q)$. To verify this fact we
note that
\[
dk_{0}\wedge dk_{1}\wedge dk_{2}\wedge dk_{3}=\left(  \det\Lambda^{-1}\right)
dy_{0}\wedge dy_{1}\wedge dy_{2}\wedge dy_{3}\text{ and }dQ\left(  k\right)
=dQ(y)
\]
under $k=\Lambda^{-1}y$. Now by (\ref{Gelfand-Leray}) and the fact that the
restriction of $\lambda_{t}$ to $V_{t}$ is unique we have $\lambda_{t}\left(
k\right)  =\left(  \det\Lambda^{-1}\right)  \lambda_{t}\left(  y\right)  $ on
$V_{t}$, i.e. $d\lambda_{t}\left(  k\right)  =d\lambda_{t}\left(  y\right)  $
under $k=\Lambda^{-1}y$ on $V_{t}$.
\end{proof}

\subsubsection{Some additional results on $\delta\left(  Q\left(  k\right)
-t\right)  $}

We now take $t=m^{2}$ with $m\in\mathbb{Q}_{p}^{\times}$. Notice that
$V_{m^{2}}$ has infinitely many points and that $\left(  k_{0},\boldsymbol{k}%
\right)  \in V_{m^{2}}$ if and only if $\left(  -k_{0},\boldsymbol{k}\right)
\in V_{m^{2}}$. In order to exploit this symmetry we need a `notion of
positivity' on $\mathbb{Q}_{p}$. To motivate our definitions consider
$a=p^{-n}ac\left(  a\right)  \in\mathbb{Q}_{p}^{\times}$, then $-a=p^{-n}%
ac(-a)$. Thus, changing the sign of $a$ is equivalent to changing the sign of
its angular component. On the other hand, the equation $x^{2}=a$ has two
solutions if and only if $n$ is even and $\left(  \frac{a_{-n}}{p}\right)
=1$, here $\left(  \frac{\cdot}{p}\right)  $ denotes the Legendre symbol. The
condition $\left(  \frac{a_{-n}}{p}\right)  =1$ means that the equation
$z^{2}\equiv a_{-n}\operatorname{mod}p$ has two solutions, say $\pm z_{0}$,
because $p\neq2$, with $z_{0}\in\left\{  1,\ldots,\frac{p-1}{2}\right\}  $ and
$-z_{0}\in\left\{  \frac{p+1}{2},\ldots,p-1\right\}  $.

We define%
\[
\mathbb{F}_{p}^{+}=\left\{  1,\ldots,\frac{p-1}{2}\right\}  \subset
\mathbb{F}_{p}^{\times}\text{ and }\mathbb{F}_{p}^{-}=\left\{  \frac{p+1}%
{2},\ldots,p-1\right\}  \subset\mathbb{F}_{p}^{\times}\text{.}%
\]

Motivated by the above discussion we introduce the following notion of `positivity'.

\begin{definition}
\label{positivity definition}We say that $a\in\mathbb{Q}_{p}^{\times}$ is
positive \ if $\overline{ac(a)}\in\mathbb{F}_{p}^{+}$, otherwise we declare
$a$ to be negative. We will use the notation $a>0$, in the first case, and
$a<0$ in the second case.
\end{definition}

The reader must be aware that this notion of positivity is not compatible with
the arithmetic operations on $\mathbb{F}_{p}$ neither on $\mathbb{Q}%
_{p}^{\times}$ because these fields cannot be ordered.

We now define the \textit{mass shells} as follows:%
\[
V_{m^{2}}^{+}=\left\{  \left(  k_{0},\boldsymbol{k}\right)  \in V_{m^{2}%
}:k_{0}>0\right\}  \text{ and }V_{m^{2}}^{-}=\left\{  \left(  k_{0}%
,\boldsymbol{k}\right)  \in V_{m^{2}}:k_{0}<0\right\}  .
\]
Hence
\begin{equation}
V_{m^{2}}=V_{m^{2}}^{+}%
{\textstyle\bigsqcup}
V_{m^{2}}^{-}%
{\textstyle\bigsqcup}
\left\{  \left(  k_{0},\boldsymbol{k}\right)  \in V_{m^{2}}:k_{0}=0\right\}  .
\label{partition}%
\end{equation}
Notice that
\[%
\begin{array}
[c]{ccc}%
V_{m^{2}}^{+} & \rightarrow & V_{m^{2}}^{-}\\
&  & \\
\left(  k_{0},\boldsymbol{k}\right)  & \rightarrow & \left(  -k_{0}%
,\boldsymbol{k}\right)
\end{array}
\]
is a bijection. We define
\[%
\begin{array}
[c]{cccc}%
\Pi: & \mathbb{Q}_{p}^{4} & \rightarrow & \mathbb{Q}_{p}^{3}\\
&  &  & \\
& \left(  k_{0},\boldsymbol{k}\right)  & \rightarrow & \boldsymbol{k},
\end{array}
\]
and $\Pi\left(  V_{m^{2}}^{+}\right)  =\Pi\left(  V_{m^{2}}^{-}\right)
:=U_{Q,m}$. Given $\boldsymbol{k}\in U_{Q,m}$, there are two $p$-adic numbers,
$k_{0}>0$ and $-k_{0}<0$, such that $\left(  k_{0},\boldsymbol{k}\right)  $,
$\left(  -k_{0},\boldsymbol{k}\right)  \in V_{m^{2}}$, thus we can define the
following two functions:%

\[%
\begin{array}
[c]{ccc}%
U_{Q,m} & \rightarrow & \mathbb{Q}_{p}^{\times}\\
&  & \\
\boldsymbol{k} & \rightarrow & \sqrt{\boldsymbol{k}\cdot\boldsymbol{k}+m^{2}%
}=:k_{0}%
\end{array}
\text{,\ \ }%
\begin{array}
[c]{ccc}%
U_{Q,m} & \rightarrow & \mathbb{Q}_{p}^{\times}\\
&  & \\
\boldsymbol{k} & \rightarrow & -\sqrt{\boldsymbol{k}\cdot\boldsymbol{k}+m^{2}%
}=:-k_{0}.
\end{array}
\]

Furthermore, we obtain the following description of the sets $V_{m^{2}}^{\pm}$:%

\begin{equation}
V_{m^{2}}^{\pm}=\left\{  \left(  k_{0},\boldsymbol{k}\right)  \in
\mathbb{Q}_{p}^{4}:k_{0}=\pm\sqrt{\boldsymbol{k}\cdot\boldsymbol{k}+m^{2}%
}\text{, for }\boldsymbol{k}\in U_{Q,m}\right\}  . \label{V_m}%
\end{equation}

\begin{lemma}
\label{lemma2}With the above notation the following assertions hold:

(i) $U_{Q,m}$ is an open subset of $\mathbb{Q}_{p}^{3}$;

(ii) the functions $\pm\sqrt{\boldsymbol{k}\cdot\boldsymbol{k}+m^{2}}$ are
$p$-adic analytic on $U_{Q,m}$;

(iii) $U_{Q,m}$ is $p$-adic bianalytic equivalent to each $V_{m^{2}}^{\pm}$,
and $V_{m^{2}}^{\pm}$ are open subsets of $\mathbb{Q}_{p}^{3}$;

(iv)%
\begin{align*}
\lambda_{m^{2}}\left(  k_{0},\boldsymbol{k}\right)   &  \mid_{V_{m^{2}}^{\pm}%
}=\frac{dk_{1}\wedge dk_{2}\wedge dk_{3}}{\pm2\sqrt{\boldsymbol{k}%
\cdot\boldsymbol{k}+m^{2}}}\mid_{U_{Q,m}}\text{ and}\\
d\lambda_{m^{2}}\left(  k_{0},\boldsymbol{k}\right)   &  \mid_{V_{m^{2}}^{\pm
}}=\frac{d^{3}\boldsymbol{k}}{\left\vert \sqrt{\boldsymbol{k}\cdot
\boldsymbol{k}+m^{2}}\right\vert _{p}}\mid_{U_{Q,m}},
\end{align*}
where $d^{3}\boldsymbol{k}$ is the normalized Haar measure of $\mathbb{Q}%
_{p}^{3}$;

(v) $%
{\displaystyle\int\limits_{\left\{  \left(  k_{0},\boldsymbol{k}\right)  \in
V_{m^{2}}:k_{0}=0\right\}  }}
\varphi\left(  k_{0},\boldsymbol{k}\right)  d\lambda_{m^{2}}\left(
k_{0},\boldsymbol{k}\right)  =0$ for any $\varphi\left(  k_{0},\boldsymbol{k}%
\right)  \in\mathbf{S}\left(  \mathbb{Q}_{p}^{4}\right)  $.
\end{lemma}

\begin{proof}
Take a point $\left(  k_{0},\boldsymbol{k}\right)  \in V_{m^{2}}^{+}$, \ then
$\left(  k_{0},\boldsymbol{k}\right)  \in V_{t}^{\left(  j\right)  }$\ for
some $j$, see (\ref{Mj}), thus there exist an open compact subset $U_{+}%
=U_{+}^{\prime}\times U_{+}^{\prime\prime}$ containing $\left(  k_{0}%
,\boldsymbol{k}\right)  $ and a $p$-adic analytic function $h_{+}%
:U_{+}^{\prime\prime}\rightarrow U_{+}^{\prime}$ such that%
\begin{equation}
V_{m^{2}}^{+}\cap U=\left\{  \left(  k_{0},\boldsymbol{k}\right)  \in
U_{+}:k_{0}=h_{+}\left(  \boldsymbol{k}\right)  \text{ with }\boldsymbol{k}\in
U_{+}^{\prime\prime}\right\}  . \label{V_m_2}%
\end{equation}

Now by (\ref{V_m}), we have
\[
h_{+}\left(  \boldsymbol{k}\right)  \mid_{U_{+}^{\prime\prime}}=\sqrt
{\boldsymbol{k}\cdot\boldsymbol{k}+m^{2}}\mid_{U_{+}^{\prime\prime}},
\]
which implies that $\sqrt{\boldsymbol{k}\cdot\boldsymbol{k}+m^{2}}$ is a
$p$-adic analytic function on $U_{Q,m}$, which is an open subset of
$\mathbb{Q}_{p}^{3}$ since it is the union of all the $U_{+}^{\prime\prime}$
which are open. In this way we establish (i)-(ii).

We now prove (iii). By (ii)
\[%
\begin{array}
[c]{ccc}%
U_{Q,m} & \rightarrow & V_{m^{2}}^{\pm}\\
&  & \\
\boldsymbol{k} & \rightarrow & \left(  \pm\sqrt{\boldsymbol{k}\cdot
\boldsymbol{k}+m^{2}},\boldsymbol{k}\right)  =:i_{\pm}\left(  \boldsymbol{k}%
\right)
\end{array}
\]
are $p$-adic bianalytic mappings, and by (i) $V_{m^{2}}^{\pm}$ are open
subsets of $\mathbb{Q}_{p}^{3}$.

The formulas (iv)-(v) follow from (\ref{Gelfand-Leray}) by a direct calculation.
\end{proof}

\begin{lemma}
\label{lemma3}(i) Each of the spaces $\boldsymbol{S}(V_{m^{2}}^{\pm})$ is
isomorphic$\ $to $\boldsymbol{S}(U_{Q,m})$ as $\mathbb{C}$-vector space.

(ii) If $\phi:U_{Q,m}\rightarrow\mathbb{C}$ is a function with compact
support, then
\[%
{\displaystyle\int\limits_{U_{Q,m}}}
\phi\left(  \boldsymbol{k}\right)  \frac{d^{3}\boldsymbol{k}}{\left\vert
\sqrt{\boldsymbol{k}\cdot\boldsymbol{k}+m^{2}}\right\vert _{p}}=%
{\displaystyle\int\limits_{V_{m^{2}}^{\pm}}}
\left(  \phi\circ i_{\pm}^{-1}\right)  \left(  \boldsymbol{k}\right)
d\lambda_{m^{2}}\left(  k_{0},\boldsymbol{k}\right)  \text{.}%
\]
(iii) If $m\in\mathbb{Q}_{p}^{\times}$ and $\varphi:\mathbb{Q}_{p}%
^{4}\rightarrow\mathbb{C}$ is a function with compact support, then%
\begin{align}%
{\displaystyle\int\limits_{^{V_{m^{2}}}}}
\varphi\left(  k\right)  d\lambda_{m^{2}}\left(  k\right)   &  =%
{\displaystyle\int\limits_{U_{Q,m}}}
\varphi\left(  \sqrt{\boldsymbol{k}\cdot\boldsymbol{k}+m^{2}},\boldsymbol{k}%
\right)  \frac{d^{3}\boldsymbol{k}}{\left\vert \sqrt{\boldsymbol{k}%
\cdot\boldsymbol{k}+m^{2}}\right\vert _{p}}\nonumber\\
&  +%
{\displaystyle\int\limits_{U_{Q,m}}}
\varphi\left(  -\sqrt{\boldsymbol{k}\cdot\boldsymbol{k}+m^{2}},\boldsymbol{k}%
\right)  \frac{d^{3}\boldsymbol{k}}{\left\vert \sqrt{\boldsymbol{k}%
\cdot\boldsymbol{k}+m^{2}}\right\vert _{p}}. \label{formula1}%
\end{align}

\end{lemma}

\begin{proof}
(i) Let $\phi_{\pm}$ be a function in $\boldsymbol{S}(V_{m^{2}}^{\pm})$, by
applying Lemma \ref{lemma2} (iii), we have $\phi_{\pm}\circ i_{\pm}%
\in\boldsymbol{S}(U_{Q,m})$. Conversely, if $\varphi\in\boldsymbol{S}%
(U_{Q,m})$, then, by Lemma \ref{lemma2} (iii), $\varphi\circ i_{\pm}^{-1}%
\in\boldsymbol{S}(V_{m^{2}}^{\pm})$. (ii) The formula follows from (i) by
applying Lemma \ref{lemma2} (iv). (iii) The formula follows from
(\ref{partition}) by applying Lemma \ref{lemma2} (iii)-(iv)-(v).
\end{proof}

\begin{remark}
\label{nota2} We set%
\[
\delta_{\pm}\left(  Q\left(  k\right)  -m^{2}\right)  :=\delta\left(  Q\left(
k\right)  -m^{2}\right)  \mid_{V_{m^{2}}^{\pm}}.
\]
If we take $\Lambda_{0}:=\left[
\begin{array}
[c]{cc}%
-1 & {\LARGE 0}\\
{\LARGE 0} & I_{3\times3}%
\end{array}
\right]  =\Lambda_{0}^{-1}\in O\left(  Q\right)  $, then
\[
\delta_{-}\left(  Q\left(  k\right)  -m^{2}\right)  =\Lambda_{0}\delta
_{+}\left(  Q\left(  k\right)  -m^{2}\right)  .
\]
Notice that instead of $\Lambda_{0}$ we can use any $\Lambda$ satisfying
$\Lambda\left(  V_{m^{2}}^{+}\right)  =V_{m^{2}}^{-}$.

Now formula (\ref{formula1}) can be written as
\[
\delta\left(  Q\left(  k\right)  -m^{2}\right)  =\delta_{+}\left(  Q\left(
k\right)  -m^{2}\right)  +\delta_{-}\left(  Q\left(  k\right)  -m^{2}\right)
.
\]

\end{remark}

\subsubsection{The $p$-adic restricted Lorentz group}

\begin{definition}
\label{Lorentz group}We define the $p$-adic restricted Lorentz group
$\mathcal{L}_{+}^{\uparrow}$ to be the largest subgroup of $SO(Q)$ such that
$\mathcal{L}_{+}^{\uparrow}\left(  V_{m^{2}}^{\pm}\right)  =V_{m^{2}}^{\pm}$.
\end{definition}

Notice that $\mathcal{L}_{+}^{\uparrow}$ is a non-trivial subgroup of $SO(Q)$.
Indeed, take $\Lambda$ in
\[
SO\left(  3\right)  =\left\{  R\in GL_{3}\left(  \mathbb{Q}_{p}\right)
:R^{T}=R^{-1}\text{, }\det R=1\right\}  ,
\]

and define%
\[
\widetilde{\Lambda}=\left[
\begin{array}
[c]{cc}%
1 & {\LARGE 0}\\
{\LARGE 0} & \Lambda
\end{array}
\right]  .
\]
Then
\[
\left(  \widetilde{\Lambda}\right)  ^{T}=\left[
\begin{array}
[c]{cc}%
1 & {\LARGE 0}\\
{\LARGE 0} & \Lambda^{-1}%
\end{array}
\right]  \text{ and }\left(  \widetilde{\Lambda}\right)  ^{T}G\widetilde
{\Lambda}=G\text{, i.e. }\widetilde{\Lambda}\in SO(Q)\text{,}%
\]
and since $\widetilde{\Lambda}\boldsymbol{k}\cdot\widetilde{\Lambda
}\boldsymbol{k}=\boldsymbol{k}\cdot\boldsymbol{k}$ we have $\widetilde
{\Lambda}\left(  V_{m^{2}}^{\pm}\right)  =V_{m^{2}}^{\pm}$.

At the moment, we do not know if $\mathcal{L}_{+}^{\uparrow}=\left\{
1\right\}  \times SO(3)$. It seems that this depends on $\mathbb{Q}_{p}$,
which can be replaced for any locally compact field of characteristic
different from $2$.

\begin{lemma}
\label{lemma4}The distributions $\delta_{\pm}\left(  Q\left(  k\right)
-m^{2}\right)  $ are invariant under $\mathcal{L}_{+}^{\uparrow}$.
\end{lemma}

\begin{proof}
Consider first $\delta_{\pm}\left(  Q\left(  k\right)  -m^{2}\right)  $. Take
$\Lambda\in\mathcal{L}_{+}^{\uparrow}$, then%
\[
\left(  \Lambda\delta_{+}\left(  Q\left(  k\right)  -m^{2}\right)
,\varphi\right)  =%
{\displaystyle\int\limits_{V_{m^{2}}^{\pm}}}
\varphi\left(  \Lambda k\right)  d\lambda_{m^{2}}\left(  k\right)  =%
{\displaystyle\int\limits_{V_{m^{2}}^{\pm}}}
\varphi\left(  k\right)  d\lambda_{m^{2}}\left(  k\right)
\]
because $\Lambda\left(  V_{m^{2}}^{+}\right)  =V_{m^{2}}^{+}$ and
$d\lambda_{m^{2}}$ is invariant under any element of $O(Q)$, see proof of
Lemma \ref{lemma1}.
\end{proof}

\section{A non-Archimedean analog of the Klein-Gordon Equation}

Given a positive real number $\alpha$ and a nonzero $p$-adic number $m$, we
define the pseudodifferential operator%
\[%
\begin{array}
[c]{ccc}%
\mathbf{S}\left(  \mathbb{Q}_{p}^{4}\right)  & \rightarrow & C\left(
\mathbb{Q}_{p}^{4}\right)  \cap L^{2}\left(  \mathbb{Q}_{p}^{4}\right) \\
&  & \\
\varphi & \rightarrow & \square_{\alpha,m}\varphi,
\end{array}
\]
where $\left(  \square_{\alpha,m}\varphi\right)  \left(  x\right)
:=\mathcal{F}_{k\rightarrow x}^{-1}\left[  \left\vert \left[  k,k\right]
-m^{2}\right\vert _{p}^{\alpha}\mathcal{F}_{x\rightarrow k}\varphi\right]  $.

We set $\mathcal{E}_{Q,m}\left(  \mathbb{Q}_{p}^{4}\right)  :=\mathcal{E}%
_{Q,m}$ to be the subspace of $\mathbf{S}^{\prime}\left(  \mathbb{Q}_{p}%
^{4}\right)  $ consisting of the distributions $T$ such that the product
$\left\vert \left[  k,k\right]  -m^{2}\right\vert _{p}^{\alpha}\mathcal{F}T$
exists in $\mathbf{S}^{\prime}\left(  \mathbb{Q}_{p}^{4}\right)  $, here
$\left\vert \left[  k,k\right]  -m^{2}\right\vert _{p}^{\alpha}$ denotes the
distribution $\varphi\rightarrow\int_{\mathbb{Q}_{p}^{4}}\left\vert \left[
k,k\right]  -m^{2}\right\vert _{p}^{\alpha}\varphi\left(  k\right)  d^{4}k$.
Notice that $\mathcal{E}\left(  \mathbb{Q}_{p}^{4}\right)  $, the space of
locally constant functions, is contained in $\mathcal{E}_{Q,m}$. We consider
$\mathcal{E}_{Q,m}$ as topological space with the topology inherited from
$\mathbf{S}^{\prime}\left(  \mathbb{Q}_{p}^{4}\right)  $.

\begin{definition}
\textit{A weak solution of}%
\begin{equation}
\square_{\alpha,m}T=S\text{, with }S\in\mathbf{S}^{\prime}\left(
\mathbb{Q}_{p}^{4}\right)  \text{,} \label{Eq1}%
\end{equation}
is a distribution $T\in\mathcal{E}_{Q,m}\left(  \mathbb{Q}_{p}^{4}\right)  $
satisfying (\ref{Eq1}).
\end{definition}

For a subset $U$ of $\mathbb{Q}_{p}^{4}$ we denote by $1_{U}$ its
characteristic function.

\begin{lemma}
\label{lemma4A}Let $T$, $S\in\mathbf{S}^{\prime}\left(  \mathbb{Q}_{p}%
^{4}\right)  $. The following assertions are equivalent:

(i) there exists $W\in\mathbf{S}^{\prime}\left(  \mathbb{Q}_{p}^{4}\right)  $
such that $TS=W;$

(ii) for each $x\in\mathbb{Q}_{p}^{4}$, there exists an open compact subset
$U$ containing $x$ so that for each each $k\in\mathbb{Q}_{p}^{4}$:%
\[
\mathcal{F}\left[  1_{U}W\right]  \left(  k\right)  :=%
{\displaystyle\int\limits_{\mathbb{Q}_{p}^{4}}}
\mathcal{F}\left[  1_{U}T\right]  \left(  l\right)  \mathcal{F}\left[
1_{U}S\right]  \left(  k-l\right)  d^{4}l
\]
exists.
\end{lemma}

\begin{proof}
Any distribution is uniquely determined by its restrictions to any countable
open covering of $\mathbb{Q}_{p}^{n}$, see e.g. \cite[p. 89]{V-V-Z}. On the
other hand, the product $TS$ exists if and only if $\mathcal{F}\left[
T\right]  \ast\mathcal{F}\left[  S\right]  $ exists, and in this case
$\mathcal{F}\left[  TS\right]  =\mathcal{F}\left[  T\right]  \ast
\mathcal{F}\left[  S\right]  $, see e.g. \cite[p. 115]{V-V-Z}. Assume that
$TS=W$ exists and take a countable covering $\left\{  U_{i}\right\}
_{i\in\mathbb{N}}$ of $\mathbb{Q}_{p}^{4}$ by open and compact subsets, then
$TS\mid_{U_{i}}=W\mid_{U_{i}}$ i.e. $1_{Ui}TS=1_{Ui}W$. We recall that \ the
product of a finite number of distributions involving at least one
distribution with compact support is associative and commutative, see e.g.
\cite[Theorem 3.19]{Taibleson}, then
\[
1_{Ui}TS=1_{Ui}\left(  1_{Ui}TS\right)  =\left(  1_{Ui}T\right)  \left(
1_{Ui}S\right)  =T\mid_{U_{i}}S\mid_{U_{i}}=W\mid_{U_{i}}.
\]

Now for each $x\in\mathbb{Q}_{p}^{4}$, there exists an open compact subset
$U_{i}$ containing $x$ such that $\mathcal{F}\left[  T\mid_{U_{i}}\right]
\ast\mathcal{F}\left[  S\mid_{U_{i}}\right]  $ $=\mathcal{F}\left[
1_{Ui}T\right]  \ast\mathcal{F}\left[  1_{Ui}S\right]  =\mathcal{F}\left[
1_{Ui}W\right]  $. Conversely, if for each $x$ there exist an open compact
subset $U_{i}$ containing $x$ (from this we get countable subcovering of
$\mathbb{Q}_{p}^{4}$ also denoted as $\left\{  U_{i}\right\}  _{i\in
\mathbb{N}}$) such that $\mathcal{F}\left[  T\mid_{U_{i}}\right]
\ast\mathcal{F}\left[  S\mid_{U_{i}}\right]  $ $=\mathcal{F}\left[
W\mid_{U_{i}}\right]  $ i.e. $T\mid_{U_{i}}S\mid_{U_{i}}=W\mid_{U_{i}}$exists,
then $TS=W$.
\end{proof}

\begin{corollary}
\label{cor1}If $TS$ exists, then supp$\left(  TS\right)  \subseteq
$supp$\left(  T\right)  \cap$supp$\left(  S\right)  $.
\end{corollary}

\begin{proof}
Since $x\notin$supp$\left(  S\right)  $, there exists a compact open set $U$
containing $x$ such $\left(  S,\varphi\right)  =0$ for any $\varphi
\in\boldsymbol{S}(U)$, hence $1_{U}S=0$, and $\mathcal{F}\left[
1_{U}T\right]  \ast\mathcal{F}\left[  1_{U}S\right]  =0=\mathcal{F}\left[
1_{U}W\right]  $, i.e. $W\mid_{U}=0$, which means $x\notin$supp$\left(
W\right)  $.
\end{proof}

\begin{remark}
Lemma \ref{lemma4A} and Corollary \ref{cor1} are valid in arbitrary dimension.
These results are well-known in the Archimedean setting, see e.g.
\cite[Theorem IX.43]{Reed-Simon}, however, such results do not appear in the
standard books of $p$-adic analysis \cite{A-K-S}, \cite{Koch},
\cite{Taibleson}, \cite{V-V-Z}.
\end{remark}

\begin{remark}
\label{nota_cuenta}(i) Let $\Omega$\ denote the characteristic function of the
interval $\left[  0,1\right]  $. Then $\Omega\left(  p^{-j}\left\Vert
x\right\Vert _{p}\right)  $ is the characteristic function of the ball
$B_{j}^{\left(  n\right)  }\left(  0\right)  $. We recall definition of the
product of two distributions. Set $\delta_{j}\left(  x\right)  :=p^{nj}%
\Omega\left(  p^{j}\left\Vert x\right\Vert _{p}\right)  $ for $j\in\mathbb{N}%
$. Given $T,S\in\mathbf{S}^{\prime}\left(  \mathbb{Q}_{p}^{n}\right)  $, their
product $TS$ is defined by%
\[
\left(  TS,\varphi\right)  =\lim_{j\rightarrow+\infty}\left(  S,\left(
T\ast\delta_{j}\right)  \varphi\right)
\]
if the limit exists for all $\varphi\in\mathbf{S}\left(  \mathbb{Q}_{p}%
^{n}\right)  $.

(ii) We assert that
\[
\left(  \left\vert \left[  k,k\right]  -m^{2}\right\vert _{p}^{\alpha
}\mathcal{F}T,\varphi\right)  =\left(  \mathcal{F}T,\left\vert \left[
k,k\right]  -m^{2}\right\vert _{p}^{\alpha}\varphi\right)
\]
for any $T\in\mathcal{E}_{Q,m}\left(  \mathbb{Q}_{p}^{4}\right)  $ and any
$\varphi\in\mathbf{S}\left(  \mathbb{Q}_{p}^{n}\right)  $. Indeed, by using
the fact that $V_{m^{2}}$ has $d^{4}y$-measure zero,
\begin{align*}
\left\vert \left[  k,k\right]  -m^{2}\right\vert _{p}^{\alpha}\ast\delta
_{j}\left(  k\right)   &  =\left(  \left\vert \left[  y,y\right]
-m^{2}\right\vert _{p}^{\alpha},\delta_{j}\left(  k-y\right)  \right) \\
&  =p^{4j}%
{\displaystyle\int\limits_{k+\left(  p^{j}\mathbb{Z}_{p}\right)  ^{4}}}
\left\vert \left[  y,y\right]  -m^{2}\right\vert _{p}^{\alpha}d^{4}y\\
&  =p^{4j}%
{\displaystyle\int\limits_{k+\left(  p^{j}\mathbb{Z}_{p}\right)
^{4}\smallsetminus V_{m^{2}}}}
\left\vert \left[  y,y\right]  -m^{2}\right\vert _{p}^{\alpha}d^{4}%
y=\left\vert \left[  k,k\right]  -m^{2}\right\vert _{p}^{\alpha}%
\end{align*}
for $j$ big enough depending on $k$. Then%
\begin{align*}
\left(  \left\vert \left[  k,k\right]  -m^{2}\right\vert _{p}^{\alpha
}\mathcal{F}T,\varphi\right)   &  =\lim_{j\rightarrow+\infty}\left(
\mathcal{F}T\left(  k\right)  ,\left[  \left\vert \left[  k,k\right]
-m^{2}\right\vert _{p}^{\alpha}\ast\delta_{j}\left(  k\right)  \right]
\varphi\left(  k\right)  \right) \\
&  =\left(  \mathcal{F}T\left(  k\right)  ,\left\vert \left[  k,k\right]
-m^{2}\right\vert _{p}^{\alpha}\varphi\left(  k\right)  \right)  .
\end{align*}

\end{remark}

\begin{lemma}
\label{lemma5}A distribution $T\in\mathcal{E}_{Q,m}\left(  \mathbb{Q}_{p}%
^{4}\right)  $ is a $\varphi$ weak solution of $\square_{\alpha,m}T=0$ if and
only if supp$\mathcal{F}T\subseteq V_{m^{2}}$.
\end{lemma}

\begin{proof}
Suppose that supp$\mathcal{F}T\subseteq V_{m^{2}}$, then by Corollary
\ref{cor1}, we have%
\[
\text{supp}\left(  \left\vert \left[  k,k\right]  -m^{2}\right\vert
_{p}^{\alpha}\mathcal{F}T\right)  \subseteq\text{supp}\left(  \mathcal{F}%
T\right)  \cap\text{supp}\left(  \left\vert \left[  k,k\right]  -m^{2}%
\right\vert _{p}^{\alpha}\right)  =\emptyset
\]
because $\left\vert \left[  k,k\right]  -m^{2}\right\vert _{p}^{\alpha
}=\left\vert \left[  k,k\right]  -m^{2}\right\vert _{p}^{\alpha}%
{\LARGE 1}_{\mathbb{Q}_{p}^{4}\smallsetminus V_{m^{2}}}$ in $\boldsymbol{S}%
^{\prime}(\mathbb{Q}_{p}^{4})$ ($V_{m^{2}}$ has $d^{4}k$-measure zero) and
supp$\left(  \left\vert \left[  k,k\right]  -m^{2}\right\vert _{p}^{\alpha
}{\LARGE 1}_{\mathbb{Q}_{p}^{4}\smallsetminus V_{m^{2}}}\right)
\subseteq\mathbb{Q}_{p}^{4}\smallsetminus V_{m^{2}}$, therefore $\left\vert
\left[  k,k\right]  -m^{2}\right\vert _{p}^{\alpha}\mathcal{F}T=0$.

Suppose now that $\left\vert \left[  k,k\right]  -m^{2}\right\vert
_{p}^{\alpha}\mathcal{F}T=0$. By contradiction, assume that supp$\mathcal{F}%
T\nsubseteq V_{m^{2}}$. Then , there exists $k_{0}\in\mathbb{Q}_{p}%
^{4}\smallsetminus V_{m^{2}}$ and an open compact subset $U\subset
\mathbb{Q}_{p}^{4}\smallsetminus V_{m^{2}}$ containing $k_{0}$ such that
$\left(  \mathcal{F}T,{\LARGE 1}_{U}\right)  \neq0$. By using Remark
\ref{nota_cuenta}-(ii) and by shrinking $U$ if necessary,%
\begin{align*}
\left(  \left\vert \left[  k,k\right]  -m^{2}\right\vert _{p}^{\alpha
}\mathcal{F}T,{\LARGE 1}_{U}\right)   &  =\left(  \mathcal{F}T,\left\vert
\left[  k,k\right]  -m^{2}\right\vert _{p}^{\alpha}{\LARGE 1}_{U}\right) \\
&  =\left\vert \left[  k_{0},k_{0}\right]  -m^{2}\right\vert _{p}^{\alpha
}\left(  \mathcal{F}T,{\LARGE 1}_{U}\right)  \neq0,
\end{align*}
contradicting $\left\vert \left[  k,k\right]  -m^{2}\right\vert _{p}^{\alpha
}\mathcal{F}T=0$.
\end{proof}

\begin{remark}
\label{nota3} Let $\varphi\in\boldsymbol{S}\left(  \mathbb{Q}_{p}^{4}\right)
$ and let $\Lambda\in\mathcal{L}_{+}^{\uparrow}$, a Lorentz transformation. We
have
\[
\mathcal{F}\left[  \varphi\left(  \Lambda x\right)  \right]  \left(  k\right)
=\mathcal{F}\left[  \varphi\right]  \left(  \Lambda k\right)
\]
and
\[
\Lambda\mathcal{F}\left[  T\right]  =\mathcal{F}\left[  \Lambda T\right]  .
\]
Hence the Fourier transform preserves Lorentz invariance, or more generally,
the Fourier transform preserves invariance under $O\left(  Q\right)  $.
\end{remark}

\begin{proposition}
\label{Prop1} The distributions
\[
\mathcal{F}\left[  T\right]  \left(  k\right)  =a\delta_{+}\left(  Q\left(
k\right)  -m^{2}\right)  +b\delta_{-}\left(  Q\left(  k\right)  -m^{2}\right)
\text{, }a,b\in\mathbb{C}\text{,}%
\]
are weak solutions of $\square_{\alpha,m}T=0$\ invariant under $\mathcal{L}%
_{+}^{\uparrow}$.
\end{proposition}

\begin{proof}
By Remark \ref{nota3}, it is sufficient to show that $\delta_{\pm}\left(
Q\left(  k\right)  -m^{2}\right)  $ are invariant solutions of $\left\vert
\left[  k,k\right]  -m^{2}\right\vert _{p}^{\alpha}\mathcal{F}\left[
T\right]  =0$, which follows from Lemmas \ref{lemma4}-\ref{lemma5}.
\end{proof}

At this point we should mention that a similar results to Lemmas
\ref{lemma4}-\ref{lemma5} and Proposition \ref{Prop1} are valid for the
Archimedean Klein-Gordon equation, see e.g. \cite{DJager} or \cite[Chapter
IV]{DJager2}.

\section{The Cauchy Problem for the non-Archimedean Klein-Gordon Equation}

In this section we study the Cauchy problem for the $p$-adic Klein-Gordon equations.

\subsection{Twisted Vladimirov pseudodifferential operators}

Let $\mathbb{C}_{1}^{\times}$ denote the multiplicative group of complex
numbers having modulus one as before. Let $\pi_{1}:\mathbb{Z}_{p}^{\times
}\rightarrow\mathbb{C}_{1}^{\times}$ be a \textit{non-trivial multiplicative
character} of $\mathbb{Z}_{p}^{\times}$ with positive \textit{conductor} $k$,
i.e. $k$ is the smallest positive integer such that $\pi_{1}\mid
_{1+p^{k}\mathbb{Z}_{p}}=1$. Some authors call a such character a
\textit{unitary character of} $\mathbb{Z}_{p}^{\times}$. We extend $\pi_{1}$
to $\mathbb{Q}_{p}^{\times}$ by putting $\pi_{1}\left(  x\right)  :=\pi
_{1}\left(  ac\left(  x\right)  \right)  $. A \textit{quasicharacter} of
$\mathbb{Q}_{p}^{\times}$ (some \ authors use \textit{multiplicative
character}) is a continuous homomorphism from $\mathbb{Q}_{p}^{\times}$ into
$\mathbb{C}^{\times}$. Every quasicharacter has the form $\pi_{s}\left(
x\right)  =\pi_{1}\left(  x\right)  \left\vert x\right\vert _{p}^{s-1}$ for
some complex number $s$.

The distribution associated with $\pi_{s}\left(  x\right)  $ has a meromorphic
continuation to the whole complex plane given by%
\begin{equation}
\left(  \pi_{s}\left(  x\right)  ,\varphi\left(  x\right)  \right)  =%
{\displaystyle\int\limits_{\mathbb{Z}_{p}}}
\pi_{1}\left(  x\right)  \left\vert x\right\vert _{p}^{s-1}\left\{
\varphi\left(  x\right)  -\varphi\left(  0\right)  \right\}  dx+%
{\displaystyle\int\limits_{\mathbb{Q}_{p}\smallsetminus\mathbb{Z}_{p}}}
\pi_{1}\left(  x\right)  \left\vert x\right\vert _{p}^{s-1}\varphi\left(
x\right)  dx, \label{form3}%
\end{equation}
see e.g. \cite[p. 117]{V-V-Z}.

On the other hand,%
\begin{equation}
\mathfrak{F}\left[  \pi_{s}\right]  \left(  \xi\right)  =\Gamma_{p}\left(
s,\pi_{1}\right)  \pi_{1}^{-1}\left(  \xi\right)  \left\vert \xi\right\vert
_{p}^{-s}\text{, for any }s\in\mathbb{C}\text{,} \label{form4}%
\end{equation}
where
\begin{align*}
\Gamma_{p}\left(  s,\pi_{1}\right)   &  =p^{sk}a_{p,k}\left(  \pi_{1}\right)
\text{, }\\
a_{p,k}\left(  \pi_{1}\right)   &  =%
{\displaystyle\int\limits_{\mathbb{Z}_{p}^{\times}}}
\pi_{1}\left(  t\right)  \chi_{p}\left(  p^{-k}t\right)  dt\text{ and
}\left\vert a_{p,k}\left(  \pi_{1}\right)  \right\vert =p^{-\frac{k}{2}%
}\text{,}%
\end{align*}
see e.g. \cite[p. 124]{V-V-Z}.

Another useful formula is the following:%

\begin{equation}
\left(  \pi_{s}\left(  x\right)  ,\varphi\left(  x\right)  \right)  =%
{\displaystyle\int\limits_{\mathbb{Q}_{p}}}
\frac{\pi_{1}\left(  x\right)  \left\{  \varphi\left(  x\right)
-\varphi\left(  0\right)  \right\}  }{\left\vert x\right\vert _{p}^{s+1}%
}dx\text{, for }\operatorname{Re}(s)>0\text{, and }\varphi\in S\left(
\mathbb{Q}_{p}\right)  \text{.} \label{form6}%
\end{equation}
The formula follows from (\ref{form3}) by using that
\begin{equation}%
{\displaystyle\int\limits_{\mathbb{Q}_{p}\smallsetminus\mathbb{Z}_{p}}}
\frac{\pi_{1}\left(  x\right)  }{\left\vert x\right\vert _{p}^{s+1}}dx=\left(
%
{\displaystyle\sum\limits_{j=1}^{\infty}}
p^{-js}\right)
{\displaystyle\int\limits_{\mathbb{Z}_{p}^{\times}}}
\pi_{1}\left(  y\right)  dy=0. \label{form6A}%
\end{equation}
Given $\alpha>0$, we define \textit{the twisted Vladimirov operator} by%
\[
\left(  \widetilde{D}_{x}^{\alpha}\varphi\right)  \left(  x\right)
=\mathfrak{F}_{k\rightarrow x}^{-1}\left(  \pi_{1}^{-1}\left(  k\right)
\left\vert k\right\vert _{p}^{\alpha}\mathfrak{F}_{x\rightarrow k}\left(
\varphi\right)  \right)  \text{ for }\varphi\in S\left(  \mathbb{Q}%
_{p}\right)  \text{.}%
\]
Notice that
\[%
\begin{array}
[c]{ccc}%
S\left(  \mathbb{Q}_{p}\right)  & \rightarrow & C\left(  \mathbb{Q}%
_{p},\mathbb{C}\right)  \cap L^{2}\\
&  & \\
\varphi & \rightarrow & \left(  \widetilde{D}_{x}^{\alpha}\varphi\right)
\end{array}
\]
is a well-defined linear operator.

\begin{lemma}
\label{lemma6}For $\alpha>0$ and $\varphi\in S\left(  \mathbb{Q}_{p}\right)
$, the following formula holds:%
\begin{equation}
\left(  \widetilde{D}_{x}^{\alpha}\varphi\right)  \left(  x\right)  =\frac
{1}{\Gamma_{p}\left(  -\alpha,\pi_{1}\right)  }%
{\displaystyle\int\limits_{\mathbb{Q}_{p}}}
\frac{\pi_{1}\left(  y\right)  \left\{  \varphi\left(  x-y\right)
-\varphi\left(  x\right)  \right\}  }{\left\vert y\right\vert _{p}^{\alpha+1}%
}dy. \label{form7}%
\end{equation}

\end{lemma}

\begin{proof}
By using (\ref{form4}), we have%
\begin{multline*}
\left(  \widetilde{D}_{x}^{\alpha}\varphi\right)  \left(  x\right)
=\mathfrak{F}_{\xi\rightarrow x}^{-1}\left(  \pi_{1}^{-1}\left(  k\right)
\left\vert k\right\vert _{p}^{\alpha}\right)  \left(  x\right)  \ast
\varphi\left(  x\right)  =\frac{1}{\Gamma_{p}\left(  -\alpha,\pi_{1}\right)
}\pi_{-\alpha}\left(  x\right)  \ast\varphi\left(  x\right) \\
=\frac{1}{\Gamma_{p}\left(  -\alpha,\pi_{1}\right)  }\left(  \pi_{-\alpha
}\left(  y\right)  ,\varphi\left(  x-y\right)  \right) \\
=\frac{1}{\Gamma_{p}\left(  -\alpha,\pi_{1}\right)  }%
{\displaystyle\int\limits_{\mathbb{Z}_{p}}}
\frac{\pi_{1}\left(  y\right)  \left\{  \varphi\left(  x-y\right)
-\varphi\left(  x\right)  \right\}  }{\left\vert y\right\vert _{p}^{\alpha+1}%
}dy+\frac{1}{\Gamma_{p}\left(  -\alpha,\pi_{1}\right)  }%
{\displaystyle\int\limits_{\mathbb{Q}_{p}\smallsetminus\mathbb{Z}_{p}}}
\frac{\pi_{1}\left(  y\right)  \varphi\left(  x-y\right)  }{\left\vert
y\right\vert _{p}^{\alpha+1}}dx\\
=\frac{1}{\Gamma_{p}\left(  -\alpha,\pi_{1}\right)  }%
{\displaystyle\int\limits_{\mathbb{Z}_{p}}}
\frac{\pi_{1}\left(  y\right)  \left\{  \varphi\left(  x-y\right)
-\varphi\left(  x\right)  \right\}  }{\left\vert y\right\vert _{p}^{\alpha+1}%
}dy\\
+\frac{1}{\Gamma_{p}\left(  -\alpha,\pi_{1}\right)  }%
{\displaystyle\int\limits_{\mathbb{Q}_{p}\smallsetminus\mathbb{Z}_{p}}}
\frac{\pi_{1}\left(  y\right)  \left\{  \varphi\left(  x-y\right)
-\varphi\left(  x\right)  \right\}  }{\left\vert y\right\vert _{p}^{\alpha+1}%
}dy,
\end{multline*}
where we used (\ref{form6A}).
\end{proof}

Note that the right-hand side of (\ref{form7}) makes sense for a wider class
of functions. For instance, for $\widetilde{\mathcal{E}}_{\alpha}\left(
\mathbb{Q}_{p}\right)  $, the $\mathbb{C}$-vector space \ of locally constant
functions $u\left(  x\right)  $ satisfying
\[%
{\displaystyle\int\limits_{\mathbb{Q}_{p}\smallsetminus\mathbb{Z}_{p}}}
\frac{\left\vert u\left(  x\right)  \right\vert }{\left\vert x\right\vert
_{p}^{\alpha+1}}dx<\infty\text{.}%
\]

Another useful formula is the following:

\begin{lemma}
\label{lemma7}For $\alpha>0$, we have%
\[
\pi_{1}^{-1}\left(  x\right)  \left\vert x\right\vert _{p}^{\alpha}=\frac
{1}{\Gamma_{p}\left(  -\alpha,\pi_{1}\right)  }%
{\displaystyle\int\limits_{\mathbb{Q}_{p}}}
\frac{\pi_{1}\left(  y\right)  \left\{  \chi_{p}\left(  yx\right)  -1\right\}
}{\left\vert y\right\vert _{p}^{\alpha+1}}dy\text{ in\ }\boldsymbol{S}%
^{\prime}\left(  \mathbb{Q}_{p}\right)  .
\]

\end{lemma}

\begin{proof}
The formula follows from (\ref{form4}) and (\ref{form6}). The proof is a
simple variation of the proof given for the case in which $\pi_{1}$ is the
trivial character, see e.g. \cite[Proposition 2.3]{Koch}.
\end{proof}

From now on we put%
\[
\pi_{1}\left(  x\right)  :=i\left(  \frac{\overline{ac\left(  x\right)  }}%
{p}\right)  \text{ for }x\in\mathbb{Z}_{p}^{\times},
\]
where $\left(  \frac{\cdot}{p}\right)  $ denotes the Legendre symbol. Note
that $k=1$ and $\pi_{1}\left(  x\right)  \in\left\{  \pm i\right\}  $,
furthermore,%
\[
\left(  \frac{-1}{p}\right)  =\left\{
\begin{array}
[c]{ccc}%
1, & \text{if} & p-1\text{ is divisible by }4\\
&  & \\
-1 & \text{if} & p-3\text{ is divisible by }4.
\end{array}
\right.
\]

\subsection{The Cauchy Problem for the $p$-adic Klein-Gordon Equation}

In this section we take $x_{0}=t$ and $\left(  x_{0},\boldsymbol{x}\right)
=\left(  t,\boldsymbol{x}\right)  \in\mathbb{Q}_{p}\times\mathbb{Q}_{p}^{3}$.
Our goal is to study the following Cauchy problem:%

\begin{equation}
\left\{
\begin{array}
[c]{llll}%
\left(  \square_{\alpha,m}u\right)  \left(  t,\boldsymbol{x}\right)  =J\left(
t,\boldsymbol{x}\right)  , &  & J\left(  t,\boldsymbol{x}\right)  \in S\left(
\mathbb{Q}_{p}^{4}\right)  & \text{(A)}\\
&  &  & \\
u\left(  t,\boldsymbol{x}\right)  \mid_{t=0}=\psi_{0}\left(  \boldsymbol{x}%
\right)  , &  & \psi_{0}\in\boldsymbol{S}(\mathbb{Q}_{p}^{3})\text{ and
}\mathfrak{F}^{-1}\left[  \psi_{0}\right]  \in\boldsymbol{S}\left(
U_{Q,m}\right)  & \text{(B)}\\
&  &  & \\
\widetilde{D}_{x}^{\alpha}u\left(  t,\boldsymbol{x}\right)  \mid_{t=0}%
=\psi_{1}\left(  \boldsymbol{x}\right)  , &  & \psi_{1}\in\boldsymbol{S}%
(\mathbb{Q}_{p}^{3})\text{ and }\mathfrak{F}^{-1}\left[  \psi_{1}\right]
\in\boldsymbol{S}\left(  U_{Q,m}\right)  & \text{(C).}%
\end{array}
\right.  \label{Cuachy_problem}%
\end{equation}

\begin{theorem}
\label{TheoA}Assume that $p-3$ is divisible by $4$.Then Cauchy problem
(\ref{Cuachy_problem}) has a weak solution given by%
\begin{align}
u\left(  t,x\right)   &  =E_{\alpha}\left(  t,\boldsymbol{x}\right)  \ast
J\left(  t,\boldsymbol{x}\right)  +\label{Solution}\\
&
{\displaystyle\int\limits_{U_{Q,m}}}
\chi_{p}\left(  -t\sqrt{\boldsymbol{k}\cdot\boldsymbol{k}+m^{2}}%
+\boldsymbol{x}\cdot\boldsymbol{k}\right)  u_{+}\left(  \boldsymbol{k}\right)
\frac{d^{3}\boldsymbol{k}}{\left\vert \sqrt{\boldsymbol{k}\cdot\boldsymbol{k}%
+m^{2}}\right\vert _{p}}+\nonumber\\
&
{\displaystyle\int\limits_{U_{Q,m}}}
\chi_{p}\left(  t\sqrt{\boldsymbol{k}\cdot\boldsymbol{k}+m^{2}}+\boldsymbol{x}%
\cdot\boldsymbol{k}\right)  u_{-}\left(  \boldsymbol{k}\right)  \frac
{d^{3}\boldsymbol{k}}{\left\vert \sqrt{\boldsymbol{k}\cdot\boldsymbol{k}%
+m^{2}}\right\vert _{p}},\nonumber
\end{align}
where $E_{\alpha}\left(  t,\boldsymbol{x}\right)  $ is a distribution on
$\boldsymbol{S}\left(  \mathbb{Q}_{p}^{4}\right)  $ satisfying
\begin{equation}
\left\vert \left[  k,k\right]  -m^{2}\right\vert _{p}^{\alpha}\mathcal{F}_{%
\begin{array}
[c]{c}%
t\rightarrow k_{0}\\
\boldsymbol{x}\rightarrow\boldsymbol{k}%
\end{array}
}\left[  E_{\alpha}\left(  t,\boldsymbol{x}\right)  \right]  =1\text{ in
}\boldsymbol{S}^{\prime}\left(  \mathbb{Q}_{p}^{4}\right)  \text{,}
\label{funda_sol}%
\end{equation}%
\[
u_{+}\left(  \boldsymbol{k}\right)  =\frac{1}{2}\left\{  \left\vert
\sqrt{\boldsymbol{k}\cdot\boldsymbol{k}+m^{2}}\right\vert _{p}\mathfrak{F}%
^{-1}\left[  \psi_{0}\right]  \left(  \boldsymbol{k}\right)  -\frac{i\pi
_{1}\left(  \sqrt{\boldsymbol{k}\cdot\boldsymbol{k}+m^{2}}\right)
}{\left\vert \sqrt{\boldsymbol{k}\cdot\boldsymbol{k}+m^{2}}\right\vert
_{p}^{\alpha-1}}\mathfrak{F}^{-1}\left[  \psi_{1}\right]  \left(
\boldsymbol{k}\right)  \right\}  ,
\]
and
\[
u_{-}\left(  \boldsymbol{k}\right)  =\frac{1}{2}\left\{  \left\vert
\sqrt{\boldsymbol{k}\cdot\boldsymbol{k}+m^{2}}\right\vert _{p}\mathfrak{F}%
^{-1}\left[  \psi_{0}\right]  \left(  \boldsymbol{k}\right)  +\frac{i\pi
_{1}\left(  \sqrt{\boldsymbol{k}\cdot\boldsymbol{k}+m^{2}}\right)
}{\left\vert \sqrt{\boldsymbol{k}\cdot\boldsymbol{k}+m^{2}}\right\vert
_{p}^{\alpha-1}}\mathfrak{F}^{-1}\left[  \psi_{1}\right]  \left(
\boldsymbol{k}\right)  \right\}  .
\]

\end{theorem}

\begin{proof}
Like in the classical case a solution of \ (\ref{Cuachy_problem})-(A) is
computed as $u_{0}\left(  t,x\right)  +u_{1}\left(  t,x\right)  $ with
$u_{0}\left(  t,x\right)  $ a particular solution of (\ref{Cuachy_problem}%
)-(A) and $u_{1}\left(  t,x\right)  $ \ a general solution of
\begin{equation}
\left(  \square_{\alpha,m}u_{1}\right)  \left(  t,\boldsymbol{x}\right)  =0.
\label{homog_problem}%
\end{equation}
The existence of a fundamental solution for (\ref{Cuachy_problem})-(A), i.e. a
distribution $E_{\alpha}\left(  t,\boldsymbol{x}\right)  $ is a distribution
on $S\left(  \mathbb{Q}_{p}^{4}\right)  $ such that $E_{\alpha}\left(
t,\boldsymbol{x}\right)  \ast J(t,x)$ is a weak solution of
(\ref{Cuachy_problem})-(A), was established in (\cite{ZG1}), see also
(\cite{ZG2}). This fundamental solution satisfies (\ref{funda_sol}).

We now show that%
\begin{align}
u_{1}\left(  t,\boldsymbol{x}\right)   &  :=%
{\displaystyle\int\limits_{U_{Q,m}}}
\chi_{p}\left(  -t\sqrt{\boldsymbol{k}\cdot\boldsymbol{k}+m^{2}}%
+\boldsymbol{x}\cdot\boldsymbol{k}\right)  u_{+}\left(  \boldsymbol{k}\right)
\frac{d^{3}\boldsymbol{k}}{\left\vert \sqrt{\boldsymbol{k}\cdot\boldsymbol{k}%
+m^{2}}\right\vert _{p}}\label{Sol_homo_problem}\\
&  +%
{\displaystyle\int\limits_{U_{Q,m}}}
\chi_{p}\left(  t\sqrt{\boldsymbol{k}\cdot\boldsymbol{k}+m^{2}}+\boldsymbol{x}%
\cdot\boldsymbol{k}\right)  u_{-}\left(  \boldsymbol{k}\right)  \frac
{d^{3}\boldsymbol{k}}{\left\vert \sqrt{\boldsymbol{k}\cdot\boldsymbol{k}%
+m^{2}}\right\vert _{p}},\nonumber
\end{align}
is a weak solution of (\ref{homog_problem}). By applying Lemmas \ref{lemma2}%
-\ref{lemma3}, we have%
\begin{align*}
&
{\displaystyle\int\limits_{U_{Q,m}}}
\chi_{p}\left(  \mp t\sqrt{\boldsymbol{k}\cdot\boldsymbol{k}+m^{2}%
}+\boldsymbol{x}\cdot\boldsymbol{k}\right)  u_{\pm}\left(  \boldsymbol{k}%
\right)  \frac{d^{3}\boldsymbol{k}}{\left\vert \sqrt{\boldsymbol{k}%
\cdot\boldsymbol{k}+m^{2}}\right\vert _{p}}\\
&  =%
{\displaystyle\int\limits_{V_{m^{2}}^{\pm}}}
\chi_{p}\left(  \left[  \left(  k_{0},\boldsymbol{k}\right)  ,\left(
-t,-\boldsymbol{x}\right)  \right]  \right)  u_{\pm}\left(  i_{\pm}%
^{-1}\left(  \boldsymbol{k}\right)  \right)  d\lambda_{m^{2}}\left(
k_{0},\boldsymbol{k}\right) \\
&  =\mathcal{F}_{\left(  k_{0},\boldsymbol{k}\right)  \rightarrow\left(
t,\boldsymbol{x}\right)  }^{-1}\left[  \left(  u_{\pm}\circ i_{\pm}%
^{-1}\right)  \left(  k\right)  \delta_{\pm}\left(  Q\left(  k\right)
-m^{2}\right)  \right]  ,
\end{align*}
whence
\begin{align*}
\mathcal{F}_{\left(  t,\boldsymbol{x}\right)  \rightarrow\left(
k_{0},\boldsymbol{k}\right)  }\left[  u_{1}\left(  t,x\right)  \right]   &
=\left[  \left(  u_{+}\circ i_{+}^{-1}\right)  \left(  k\right)  \delta
_{+}\left(  Q\left(  k\right)  -m^{2}\right)  \right] \\
&  +\left[  \left(  u_{-}\circ i_{-}^{-1}\right)  \left(  k\right)  \delta
_{-}\left(  Q\left(  k\right)  -m^{2}\right)  \right]  .
\end{align*}
By Lemma \ref{lemma5}, $u_{1}\left(  t,x\right)  $ is a weak solution of
(\ref{homog_problem}), if
\[
\text{supp}\left(  u_{\pm}\circ i_{\pm}^{-1}\right)  \left(  k\right)
\delta_{\pm}\left(  Q\left(  k\right)  -m^{2}\right)  \subseteq V_{m^{2}}.
\]
This last condition is verified by applying Corollary \ref{cor1} and the fact
that
\[
\text{supp}\left(  u_{\pm}\circ i_{\pm}^{-1}\right)  \subseteq V_{m^{2}}^{\pm
}\subset V_{m^{2}}\text{ and supp}\left(  \delta_{\pm}\left(  Q\left(
k\right)  -m^{2}\right)  \right)  \subseteq V_{m^{2}}^{\pm}\subset V_{m^{2}}.
\]

The verification of (\ref{Cuachy_problem})-(B) is straight forward. To\ verify
(\ref{Cuachy_problem})-(C) we proceed as follows. By using Lemma \ref{lemma6},
Fubini's theorem and Lemma \ref{lemma7}, we get the following formula:%
\begin{multline*}
\widetilde{D}_{t}^{\alpha}\left[
{\displaystyle\int\limits_{U_{Q,m}}}
\chi_{p}\left(  \mp t\sqrt{\boldsymbol{k}\cdot\boldsymbol{k}+m^{2}%
}+\boldsymbol{x}\cdot\boldsymbol{k}\right)  u_{\pm}\left(  \boldsymbol{k}%
\right)  \frac{d^{3}\boldsymbol{k}}{\left\vert \sqrt{\boldsymbol{k}%
\cdot\boldsymbol{k}+m^{2}}\right\vert _{p}}\right]  =\\%
{\displaystyle\int\limits_{U_{Q,m}}}
\frac{\chi_{p}\left(  \mp t\sqrt{\boldsymbol{k}\cdot\boldsymbol{k}+m^{2}%
}+\boldsymbol{x}\cdot\boldsymbol{k}\right)  u_{\pm}\left(  \boldsymbol{k}%
\right)  }{\left\vert \sqrt{\boldsymbol{k}\cdot\boldsymbol{k}+m^{2}%
}\right\vert _{p}}\times\\
\left\{  \frac{1}{\Gamma_{p}\left(  -\alpha,\pi_{1}\right)  }%
{\displaystyle\int\limits_{\mathbb{Q}_{p}}}
\frac{\pi_{1}\left(  y\right)  \left\{  \chi_{p}\left(  \pm y\sqrt
{\boldsymbol{k}\cdot\boldsymbol{k}+m^{2}}\right)  -1\right\}  }{\left\vert
y\right\vert _{p}^{\alpha+1}}dy\right\}  d^{3}\boldsymbol{k=}\\
\pi_{1}^{-1}\left(  \pm1\right)
{\displaystyle\int\limits_{U_{Q,m}}}
\chi_{p}\left(  \mp t\sqrt{\boldsymbol{k}\cdot\boldsymbol{k}+m^{2}%
}+\boldsymbol{x}\cdot\boldsymbol{k}\right)  \pi_{1}^{-1}\left(  \sqrt
{\boldsymbol{k}\cdot\boldsymbol{k}+m^{2}}\right)  \times\\
\left\vert \sqrt{\boldsymbol{k}\cdot\boldsymbol{k}+m^{2}}\right\vert
_{p}^{\alpha-1}u_{\pm}\left(  \boldsymbol{k}\right)  d^{3}\boldsymbol{k}.
\end{multline*}

Now Condition (\ref{Cuachy_problem})-(C) follows from the previous formula.
\end{proof}

\begin{remark}
(i) Note that the condition `$p-3$ is divisible by $4$' is required only to
establish (\ref{Cuachy_problem})-(C). At the moment, we do not know if this
condition is necessary to have (\ref{Cuachy_problem})-(C).

\noindent(ii) The parameter $\alpha$ does not have any influence on the
solutions of (\ref{homog_problem}).
\end{remark}

Like in the Archimedean case the non-Archimedean Klein-Gordon equations admit
plane waves.

\begin{lemma}
\label{Lemma_plane_waves}Existence of plane waves. Let $\left(
E,\boldsymbol{p}\right)  \in V_{m^{2}}^{\pm}$, i.e. $E=\pm\sqrt{\boldsymbol{p}%
\cdot\boldsymbol{p}+m^{2}}$. Then $u\left(  t,\boldsymbol{x}\right)  =\chi
_{p}\left(  \left[  \left(  t,\boldsymbol{x}\right)  ,\left(  E,\boldsymbol{p}%
\right)  \right]  \right)  $ is a weak solution of $\left(  \square_{\alpha
,m}u\right)  \left(  t,\boldsymbol{x}\right)  =0$.
\end{lemma}

\begin{proof}
(i) Since $\mathcal{F}_{\left(  t,\boldsymbol{x}\right)  \rightarrow\left(
k_{0},\boldsymbol{k}\right)  }\left[  u\left(  t,\boldsymbol{x}\right)
\right]  =\delta\left(  k_{0}-E,\boldsymbol{k-p}\right)  $, the results
follows from Lemma \ref{lemma5}.
\end{proof}

\section{Further Results on the $p$-adic Klein-Gordon Equation}

In this section we change the notation slightly, this facilitates the
comparison with the classical results and constructions.

Set
\[
\omega\left(  \boldsymbol{k}\right)  :=\sqrt{\boldsymbol{k}\cdot
\boldsymbol{k}+m^{2}}\text{ for }\boldsymbol{k}\in U_{Q,m}.
\]

The function $\omega\left(  \boldsymbol{k}\right)  $ is a $p$-adic analytic
function on $U_{Q,m}$, cf. Lemma \ref{lemma2}, and $\omega\left(
\boldsymbol{k}\right)  \neq0$ for any $\boldsymbol{k}\in U_{Q,m}$. Then, by
Taylor formula, $\left\vert \omega\left(  \boldsymbol{k}\right)  \right\vert
_{p}$ is a locally constant function on $U_{Q,m}$, and if $\phi_{\pm}%
\in\boldsymbol{S}\left(  U_{Q,m}\right)  $, then $\left\vert \omega\left(
\boldsymbol{k}\right)  \right\vert _{p}^{\pm1}\phi_{\pm}\left(  \boldsymbol{k}%
\right)  \in\boldsymbol{S}\left(  U_{Q,m}\right)  $.

Then
\begin{equation}
\phi\left(  t,\boldsymbol{x}\right)  =%
{\displaystyle\int\limits_{U_{Q,m}}}
\chi_{p}\left(  -\boldsymbol{x}\cdot\boldsymbol{k}\right)  \left\{  \chi
_{p}\left(  -t\omega\left(  \boldsymbol{k}\right)  \right)  \phi_{+}\left(
\boldsymbol{k}\right)  +\chi_{p}\left(  t\omega\left(  \boldsymbol{k}\right)
\right)  \phi_{-}\left(  \boldsymbol{k}\right)  \right\}  d^{3}\boldsymbol{k}
\label{formula_phi}%
\end{equation}
is a weak solution of $\left(  \square_{\alpha,m}\phi\right)  \left(
t,\boldsymbol{x}\right)  =0$ for any $\phi_{\pm}\in\boldsymbol{S}\left(
U_{Q,m}\right)  $. Note that if we replace $\chi_{p}\left(  -\boldsymbol{x}%
\cdot\boldsymbol{k}\right)  $ by $\chi_{p}\left(  \boldsymbol{x}%
\cdot\boldsymbol{k}\right)  $ in (\ref{formula_phi}) we get another weak solution.

As in the classical case, for the quantum interpretation we specialize to
`positive energy solutions' which have $\phi_{-}\left(  \boldsymbol{k}\right)
=0$. Then the solution is determined by a single complex valued initial
condition. The solution (\ref{formula_phi}) with initial condition%

\[
\Psi\in L_{Q,m}^{2}:=L^{2}\left(  \left\{  \Phi\in L^{2}\left(  \mathbb{Q}%
_{p}^{3}\right)  :\text{supp}\left(  \mathfrak{F}\Phi\right)  \subseteq
U_{Q,m}\right\}  ,d^{3}\boldsymbol{k}\right)
\]
-where the condition `supp$\left(  \mathfrak{F}\Phi\right)  \subseteq U_{Q,m}%
$' means \ that there exists a function $\Phi^{\prime}$ in the equivalence
class containing $\mathfrak{F}\Phi$ such that supp$\left(  \Phi^{\prime
}\right)  \subseteq U_{Q,m}$- is given by%
\begin{equation}
\Psi\left(  t,\boldsymbol{x}\right)  =%
{\displaystyle\int\limits_{U_{Q,m}}}
\chi_{p}\left(  -t\omega\left(  \boldsymbol{k}\right)  -\boldsymbol{x}%
\cdot\boldsymbol{k}\right)  \left(  \mathfrak{F}\Psi\right)  \left(
\boldsymbol{k}\right)  d^{3}\boldsymbol{k}. \label{wave_function}%
\end{equation}

Set%
\[%
\begin{array}
[c]{llll}%
U\left(  t\right)  : & L_{Q,m}^{2} & \rightarrow & L^{2}\left(  \mathbb{Q}%
_{p}^{3}\right) \\
&  &  & \\
& \Psi & \rightarrow & \Psi\left(  t,\boldsymbol{x}\right)  =\mathfrak{F}%
_{\boldsymbol{k}\rightarrow\boldsymbol{x}}^{-1}\left(  \chi_{p}\left(
-t\omega\left(  \boldsymbol{k}\right)  \right)  \mathfrak{F}_{\boldsymbol{x}%
\rightarrow\boldsymbol{k}}\Psi\right)  ,
\end{array}
\]
for $t\in\mathbb{Q}_{p}$.

\begin{lemma}
\label{lemma8}(i) $U\left(  t\right)  $, $t\in\mathbb{Q}_{p}$ is a group of
unitary operators on $L_{Q,m}^{2}$.

(ii) $st.-\lim_{t\rightarrow0}U\left(  t\right)  =I$.
\end{lemma}

\begin{proof}
It is a straightforward calculation.
\end{proof}

We now consider the effects of space-time translations by $a=\left(
a_{0},\boldsymbol{a}\right)  \in\mathbb{Q}_{p}^{4}$ and rotations $R$ in
$SO\left(  3\right)  $. The transformations $\left\{  a,R\right\}  $ act on
$\mathbb{Q}_{p}^{4}$ naturally and they form a group, the semi-direct product
$\mathbb{Q}_{p}^{4}\ltimes SO\left(  3\right)  $, with group law given by%
\[
\left\{  a,R\right\}  \left\{  a^{\prime},R^{\prime}\right\}  =\left\{
a+Ra^{\prime},RR^{\prime}\right\}  .
\]

We attach to each $\left\{  a,R\right\}  $ the operator
\[
\left(  U_{0}\left(  a,R\right)  \psi\right)  \left(  \boldsymbol{x}\right)
:=\mathfrak{F}_{\boldsymbol{k}\rightarrow\boldsymbol{x}}^{-1}\left[  \chi
_{p}\left(  \left[  \left(  a_{0},\boldsymbol{a}\right)  ,\left(
\omega\left(  \boldsymbol{k}\right)  ,\boldsymbol{k}\right)  \right]  \right)
\left(  \mathfrak{F}_{\boldsymbol{k}\rightarrow\boldsymbol{x}}\psi\right)
\left(  R^{-1}\boldsymbol{k}\right)  \right]
\]
for $\psi\in L_{Q,m}^{2}$.

\begin{lemma}
(i) The correspondence $\left\{  a,R\right\}  \rightarrow U_{0}\left(
a,R\right)  $ gives rise a unitary representation of $\mathbb{Q}_{p}%
^{4}\ltimes SO\left(  3\right)  $\ in $L_{Q,m}^{2}$.

(ii) For the wave functions (\ref{wave_function}), we have%
\[
\left(  U_{0}\left(  a,R\right)  \Psi\left(  t,\boldsymbol{\cdot}\right)
\right)  \left(  \boldsymbol{x}\right)  =\Psi\left(  t-a_{0},R^{-1}\left(
\boldsymbol{x}-\boldsymbol{a}\right)  \right)  .
\]

\end{lemma}

\begin{proof}
(i) The calculations involved are similar to the ones required in the proof of
Lemma \ref{lemma8}-(i). (ii) It is a straightforward calculation.
\end{proof}

\begin{acknowledgement}
The author wishes to thank to the referee for his/her careful reading of the
original manuscript.
\end{acknowledgement}

\end{document}